\documentclass[11pt]{article}

\usepackage[leqno]{amsmath}
\usepackage{amssymb}
\usepackage{amsthm}
\usepackage{enumerate}
\usepackage{verbatim}
\usepackage{mathrsfs}
\usepackage{dsfont}
\usepackage{graphicx}
\usepackage{soul}

\newcommand{\tu}{\textup}
\newcommand{\ds}{\displaystyle}

\newcommand{\mr}{\mathrm}

\newcommand{\mf}{\mathfrak}
\newcommand{\mc}{\mathcal}
\newcommand{\mb}{\mathbf}

\newcommand{\scr}{\mathscr}

\newcommand{\bihom}{\rightleftarrows}

\newcommand{\eps}{\varepsilon}
\newcommand{\wt}{\widetilde}
\newcommand{\wh}{\widehat}
\newcommand{\BAR}[1]{\overline{#1}}

\newcommand{\defiff}{\stackrel{\text{\tiny{def}}}{\iff}}

\usepackage{mathtools}
\newcommand{\defeq}{\vcentcolon=}

\newcommand{\fieldfont}[1]{\mathbb{#1}}
\newcommand{\N}{\fieldfont{N}}

\newcommand{\FO}{\mr{FO}}
\newcommand{\Los}{{\L}o\'s}

\newtheoremstyle{theorem-style}
  {}
  {}
  {\slshape}
  {}
  {\bf}
  {.}
  {.5em}
  {}

\newtheoremstyle{claim-style}
  {}
  {}
  {\it}
  {}
  {}
  {.}
  {.5em}
  {}

\newtheorem{thm}{Theorem}

\newtheorem{prop}[thm]{Proposition}
\newtheorem{la}[thm]{Lemma}

\newtheorem{cor}[thm]{Corollary}

\theoremstyle{claim-style}

\newtheorem{claim}[thm]{\underline{Claim}}

\theoremstyle{definition}
\newtheorem{df}[thm]{Definition}
\newtheorem{rmk}[thm]{Remark}

\theoremstyle{remark}

\usepackage{fullpage}
\usepackage{stmaryrd}
\usepackage[stable]{footmisc}

\usepackage[draft=false,hypertexnames=false,bookmarksopenlevel=1,bookmarksopen,bookmarksnumbered,colorlinks,plainpages=false,linktocpage,linkcolor=black,citecolor=black]{hyperref}

\newcommand{\lemp}{\le_{\mr{mp}}}

\newcommand{\SUB}[1]{\mr{SUB}_{#1}}
\newcommand{\MODEL}{\mr{MODEL}}

\newcommand{\MinStruct}{\mr{MinCores}}

\newcommand{\A}{\mc{A}}

\newcommand{\ACzero}{\mr{AC}^0}
\newcommand{\NCone}{\mr{NC}^1}

\newcommand{\tw}{\mb{tw}}
\newcommand{\td}{\mb{td}}
\newcommand{\lp}{\mb{lp}}

\newcommand{\Gaif}{\mr{Gaif}}

\newcommand{\retr}{\stackrel{\scriptscriptstyle\supseteq}{\to}}

\renewcommand{\mf}{\mc}

\begin{document}

\title{An Improved Homomorphism Preservation Theorem\\
From Lower Bounds in Circuit Complexity}

\date{September 18, 2016}
        
\author{Benjamin Rossman\thanks{Supported by NSERC and the JST ERATO Kawarabayashi Large Graph Project. This paper was partially written at the National Institute of Informatics in Tokyo and during a visit to IMPA, the National Institute for Pure and Applied Mathematics in Rio de Janeiro.}\\ University of Toronto}

\maketitle

\begin{abstract}
Previous work of the author \cite{rossman2008homomorphism} showed that the Homomorphism Preservation Theorem of classical model theory remains valid when its statement is restricted to finite structures. In this paper, we give a new proof of this result via a reduction to lower bounds in circuit complexity, specifically on the $\ACzero$ formula size of the colored subgraph isomorphism problem. Formally, we show the following: if a first-order sentence $\Phi$ of quantifier-rank $k$ is preserved under homomorphisms on finite structures, then it is equivalent on finite structures to an existential-positive sentence $\Psi$ of quantifier-rank $k^{O(1)}$. Quantitatively, this improves the result of \cite{rossman2008homomorphism}, where the upper bound on the quantifier-rank of $\Psi$ is a non-elementary function of $k$.
\end{abstract}

\newpage

\tableofcontents

\newpage

\section{Introduction}\label{sec:introduction}

Preservation theorems are a family of results in classical model theory that equate semantic and syntactic properties of first-order formulas. A prominent example\,---\,and the subject of this paper\,---\,is the Homomorphism Preservation Theorem, which states that a first-order sentence is preserved under homomorphisms if, and only if, it is equivalent to an existential-positive sentence. (Definitions for the various terms in this theorem are given in Section \ref{sec:preservation-theorems}.) Two related classical preservation theorems are the \Los-Tarski Theorem (preserved under {\em embedding} homomorphisms $\Leftrightarrow$ equivalent to an {\em existential} sentence) and Lyndon's Theorem (preserved under {\em surjective} homomorphism $\Leftrightarrow$ equivalent to a {\em positive} sentence).

In all classical preservation theorems, the ``syntactic property $\Rightarrow$ semantic property'' direction is straightforward, while the ``semantic property $\Rightarrow$ syntactic property'' direction is typically proved by an application of the Compactness Theorem.\footnote{The Compactness Theorem states that a first-order theory $T$ (i.e.\ set of first-order sentences) is consistent (i.e.\ there exists a structure $\A$ which satisfies every sentence in $T$) if every finite sub-theory of $T$ is consistent. (See \cite{Hodges} for background and proofs of various preservation/amalgamation/interpolation theorems in classical model theory.)} In order to use compactness, it is essential that the semantic property (i.e.\ preservation under a certain relationship between structures) holds with respect to {\em all} structures, that is, both finite and infinite. One may also ask about the status of classical preservation theorems relative to a class of structures $\scr C$. So long as compactness holds in $\scr C$ (for example, whenever $\scr C$ is first-order axiomatizable), so too will all of the classical preservation theorems. The situation is less clear when $\scr C$ is the class of finite structures (or a subclass thereof), as the Compactness Theorem is easily seen to be false when restricted to finite structures.\footnote{Consider the theory $T = \{\Phi_n : n \in \N\}$ where $\Phi_n$ expresses ``there exist $\ge n$ distinct elements''. Every finite sub-theory of $T$ has a finite model, but $T$ itself does not.} 

The program of classifying theorems in classical model theory according to their validity over finite structures was a major line of research, initiated by Gurevich \cite{Gur84}, in the area known as {\em finite model theory} (see \cite{EF,FMTapps,Lib}). The status of preservation theorems in particular was systematically investigated in \cite{AlG97,RW95}. Given the failure of the Compactness Theorem on finite structures, it is not surprising that nearly all of the classical preservation theorems become false when their statements are restricted to finite structures. A counterexample of Tait \cite{Tait} from 1959 showed that the \Los-Tarski Theorem is false over finite structures, while Ajtai and Gurevich \cite{AG87} in 1987 gave the demise of Lyndon's Theorem via a stronger result in circuit complexity. Namely, they showed that $\mr{Monotone} \cap \ACzero \ne \mr{Monotone\text{-}AC^0}$, that is, there is a (semantically) monotone Boolean function that is computable by $\ACzero$ circuits, but not by (syntactically) monotone $\ACzero$ circuits. The failure of Lyndon's theorem on finite structures follows via the {\em descriptive complexity} correspondence between $\ACzero$ and first-order logic. (See \cite{Imm} about the nexus between logics and complexity classes.)

Given the failure of both the \Los-Tarski and Lyndon Theorems, it might be expected that the Homomorphism Preservation Theorem also fails over finite structures (as it seems to live at the intersection of \Los-Tarski and Lyndon). On the contrary, however, previous work of the author \cite{rossman2008homomorphism} showed that the Homomorphism Preservation Theorem remains valid over finite structures. The technique of \cite{rossman2008homomorphism} is model-theoretic: its starting point is a new compactness-free proof of the classical theorem, which is then adapted to finite structures. (A summary of the argument is included in Section~\ref{sec:comparison}.) In the present paper, we give a completely different proof of this result\,---\,and moreover obtain a quantitative improvement\,---\,via a reduction to lower bounds in circuit complexity. In particular, we rely on a recent result (of independent interest) that the $\ACzero$ formula size of the colored $G$-subgraph isomorphism problem is $n^{\Omega(\textbf{tree-depth}(G)^\eps)}$ for an absolute constant $\eps > 0$.

\paragraph{Related Work.}
Prior to \cite{rossman2008homomorphism}, the status of the Homomorphism Preservation Theorem on finite structures was investigated by Feder and Vardi \cite{FV03}, Gr\"aedel and Rosen \cite{GR99}, and Rosen \cite{Rosen95}, who resolved special cases of the question for restricted classes of first-order sentences. Another special case is due to Atserias \cite{atserias2008digraph} in the context of CSP dualities. (See \cite{rossman2008homomorphism} for a discussion of these results.) A different\,---\,and incomparable\,---\,line of results \cite{ADK06,dawar2010homomorphism,nesetril2014first} proves versions of the Homomorphism Preservation Theorems restricted to various {\em sparse} classes of finite structures (see Ch.\ 10 of \cite{nevsetril2012sparsity}, as well as \cite{atserias2008preservation} related to the \Los-Tarski Theorem). See Stolboushkin \cite{Stol95} for an alternative counterexample showing that Lyndon's Theorem fails on finite structures, which is simpler than Ajtai and Gurevich \cite{AG87} (but doesn't extend to show $\mr{Monotone\text{-}AC^0} \ne \mr{Monotone} \cap \mr{AC^0}$).

\paragraph{Outline.}
The rest of the paper is organized as follows. Because our narrative jumps between logic, graph theory and circuit complexity, for readability sake the various preliminaries\,---\,which may be familiar (at least in part) to many readers\,---\,are presented in separate sections as needed throughout the paper. In Section \ref{sec:prelims1}, we review basic definitions related to structures, homomorphisms, and first-order logic. In Section \ref{sec:preservation-theorems}, we formally state the various preservation theorems discussed in the introduction, including our main result (Theorem \ref{thm:improved}). Section \ref{sec:prelims2} includes the necessary background on circuit complexity ($\ACzero$ and monotone projections) and graph theory (tree-width, tree-depth, and minor-monotonicity). In Section \ref{sec:subgraph-isomorphism}, we introduce the {colored $G$-subgraph isomorphism} problem and state the known bounds on its complexity for $\ACzero$ circuits and $\ACzero$ formulas. Section \ref{sec:prelims3} states a needed lemma from descriptive complexity ($\mr{FO} = \ACzero$) and a result connecting quantifier-rank to tree-depth. In Section \ref{sec:improved}, we prove our main result (Theorem \ref{thm:improved}) via a reduction to lower bounds for colored $G$-subgraph isomorphism. (After all the preliminaries, the reduction itself is relatively simple.) For comparison sake, the previous model-theoretic proof technique of \cite{rossman2008homomorphism} is summarized in Section \ref{sec:comparison}. We conclude in Section \ref{sec:conclusion} with a brief discussion of syntax vs.\ semantics in circuit complexity.

\section{Preliminaries, I}\label{sec:prelims1}

\subsection{Structures and Homomorphisms}

Throughout this paper, let $\sigma$ be a fixed finite relational signature, that is, a list of relation symbols $R^{(r)}$ (where $r \in \N$ denotes the arity of $R$). A {\em structure} $\mf A$ consists of a set $A$ (called the unvierse of $\mf A$) together with interpretations $R^{\mf A} \subseteq A^r$ for each relation symbol $R^{(r)}$ in $\sigma$. A priori, structures may be finite or infinite.

A {\em homomorphism} from a structure $\mf A$ to a structure $\mf B$ is a map $f : A \to B$ such that $(a_1,\dots,a_r) \in R^{\mf A} \Longrightarrow (f(a_1),\dots,f(a_r)) \in R^{\mf B}$ for every $R^{(r)} \in \sigma$ and $(a_1,\dots,a_r) \in A^r$. Notation $\mf A \to \mf B$ asserts the existence of a homomorphism from $\mf A$ to $\mf B$. 

A homomorphism $f : \mf A \to \mf B$ is an {\em embedding} if $f$ is one-to-one and satisfies $(a_1,\dots,a_r) \in R^{\mf A} \iff (f(a_1),\dots,f(a_r)) \in R^{\mf B}$ for every $R^{(r)} \in \sigma$ and $(a_1,\dots,a_r) \in A^r$.

\subsection{First-Order Logic}

{\em First-order formulas} (in the relational signature $\sigma$) are constructed out of atomic formulas (of the form $x_1 = x_2$ or $R(x_1,\dots,x_r)$ where $R^{(r)} \in \sigma$ and $x_i$'s are variables) via boolean connectives ($\varphi \wedge \psi$, $\varphi \vee \psi$, and $\neg\varphi$) and universal and existential quantification ($\forall x\ \varphi(x)$ and $\exists x\ \varphi(x)$). For a structure $\mf A$ and a first-order formula $\varphi(x_1,\dots,x_k)$ and a tuple of elements $\vec a \in A^k$, notation $\mf A \models \varphi(\vec a)$ is the statement that $\mf A$ satisfies $\varphi$ with $\vec a$ instantiating the free variables $\vec x$. First-order formulas with no free variables are called {\em sentences} and represented by capital Greek letters $\Phi$ and $\Psi$.

A first-order sentence (or formula) is said to be:
\begin{itemize}
\item
{\em positive} if it does not contain any negations (that is, it has no sub-formula of the form $\neg\varphi$),
\item
{\em existential} if it contains only existential quantifiers (that is, it has no universal quantifiers) and has no negations outside the scope of any quantifier, and
\item
{\em existential-positive} if it is both existential and positive.
\end{itemize}

Two important parameters first-order sentences are quantifier-rank and variable-width. {\em Quantifier-rank} is the maximum nesting depth of quantifiers. {\em Variable-width} is the maximum number of free variables in a sub-formula. As we will see in Section \ref{sec:prelims3}, under the descriptive complexity characterization of first-order logic in terms of $\ACzero$ circuits, variable-width corresponds to $\ACzero$ circuit size and quantifier-rank corresponds to $\ACzero$ formula size (or, more accurately, $\ACzero$ formula depth when fan-in is restricted to $O(n)$).

Note that first-order sentences are not assumed to be in prenex form. For example, the formula $(\exists x\ P(x)) \vee (\exists y\ \neg Q(y))$ is existential (but not positive) and has quantifier-rank~$1$ and variable-width~$1$.

\section{The Homomorphism Preservation Theorem}\label{sec:preservation-theorems}

\begin{df}
A first-order sentence $\Phi$ is {\em preserved under homomorphisms \tu[on finite structures\tu]} if $(\mf A \models \Phi$ and $\mf A \to \mf B) \Longrightarrow \mf B \models \Phi$ for all [finite] structures $\mf A$ and $\mf B$. The notions of {\em preserved under embeddings} and {\em preserved under surjective homomorphisms} are defined similarly.
\end{df}

We now formally state the three classical preservations mentioned in the introduction.

\begin{thm}[\Los-Tarski\,/\,Lyndon\,/\,Homomorphism Preservation Theorems \cite{Ly59}]\label{thm:preservation}
\item
A first-order sentence is preserved under \tu[embedding\,\tu/\,surjective\,\tu/\,all\tu] homomorphisms if, and only if, it is equivalent to an \tu[existential\,\tu/\,positive\,\tu/\,existential-positive\tu] sentence.
\end{thm}

As discussed in the introduction, \Los-Tarski and Lyndon's Theorems become false when restricted to finite structures.

\begin{thm}[Failure of \Los-Tarski and Lyndon Theorems on Finite Structures \cite{AG87,Tait}]\label{thm:failure}
\item
There exists a first-order sentence that is preserved under \tu[embedding\,\tu/\,surjective\tu] homomorphisms on finite structures, but is not equivalent on finite structures to any \tu[existential\,\tu/\,positive\tu] sentence.
\end{thm}

In contrast, the Homomorphism Preservation Theorem remains valid over finite structures. 

\begin{thm}[Homomorphism Preservation Theorem on Finite Structures \cite{rossman2008homomorphism}]\label{thm:finiteHPT}
\item
If a first-order sentence of quantifier-rank $k$ is preserved under homomorphisms on finite structures, then it is equivalent on finite structures to an existential-positive sentence of quantifier-rank $\beta(k)$, for some computable function $\beta : \N \to \N$.
\end{thm}

We will refer to $\beta : \N \to \N$ in Theorem \ref{thm:finiteHPT} as the ``quantifier-rank blow-up''. (Formally, there is one computable function $\beta_\sigma : \N \to \N$ for each finite relational signature $\sigma$.) We remark that the upper bound on $\beta(k)$ given by the proof of Theorem \ref{thm:finiteHPT} is a non-elementary function of $k$ (i.e.\ it is grows faster than any bounded-height tower of exponentials). In contrast, a second result in \cite{rossman2008homomorphism} shows that the optimal bound $\beta(k)=k$ holds in the classical Homomorphism Preservation Theorem.

\begin{thm}[``Equi-rank'' Homomorphism Preservation Theorem \cite{rossman2008homomorphism}]\label{thm:equirank}
\item
If a first-order sentence of quantifier-rank $k$ is preserved under homomorphism, then it is equivalent to an existential-positive sentence of quantifier-rank $k$.
\end{thm}

Due to reliance on the Compactness Theorem, the original proof of the classical Homomorphism Preservation Theorem gives no computable upper bound whatsoever on the quantifier-rank blow-up. Theorem \ref{thm:equirank} is proved by a constructive, compactness-free argument (see Section \ref{sec:comparison}). In \cite{rossman2008homomorphism} I conjectured that this stronger ``equi-rank'' theorem is valid over finite structures. However, new techniques were clearly needed to improve the non-elementary upper bound on $\beta(k)$.

The main result of the present paper is a completely new proof of Theorem \ref{thm:finiteHPT}, which moreover gives a polynomial upper bound on $\beta(k)$.

\begin{thm}[``Poly-rank'' Homomorphism Preservation Theorem on Finite Structures]\label{thm:improved}
\item
If a first-order sentence of quantifier-rank $k$ is preserved under homomorphisms on finite structures, then it is equivalent on finite structures to an existential-positive sentence of quantifier-rank $k^{O(1)}$.
\end{thm}

The proof of Theorem \ref{thm:improved} involves a reduction to the $\ACzero$ formula size of $\SUB{G}$, the {colored $G$-subgraph isomorphism} problem. This reduction transforms lower bounds on the $\ACzero$ formula size of $\SUB{G}$ into upper bounds on the quantifier-rank blow-up $\beta(k)$ in Theorem \ref{thm:finiteHPT}. In Section \ref{sec:first-bound}, we derive an exponential upper bound $\beta(k) \le 2^{O(k)}$ from an existing lower bound of  \cite{rossman2014formulas} on the $\ACzero$ formula size of $\SUB{P_k}$ (also known as the {distance-$k$ connectivity} problem). Two further steps, described in Section \ref{sec:extension}, are required for the polynomial upper bound $\beta(k) \le k^{O(1)}$ of Theorem \ref{thm:improved}. The first is a new result in graph minor theory from \cite{RossmanTREEDEPTH} (joint work with Ken-ichi Kawarabayashi), which gives a ``polynomial excluded-minor approximation'' of tree-depth, analogous to the Polynomial Grid-Minor Theorem of Chekuri and Chuzhoy \cite{chekuri2014polynomial}. The second ingredient, in a forthcoming paper of the author \cite{RossmanSUBGRAPH}, is a lower bound on $\ACzero$ formula size of $\SUB{G}$ in the special case where $G$ is complete binary tree.

\section{Preliminaries, II}\label{sec:prelims2}

\subsection{Circuit Complexity}

We consider {\em Boolean circuits} with unbounded fan-in $\mr{AND}$ and $\mr{OR}$ gates and negations on inputs. That is, inputs are labelled by variables $x_i$ or negated variables $\BAR x_i$ (where $i$ comes from some finite index set, typically $\{1,\dots,n\}$). We measure {\em size} by the number of gates and {\em depth} by the maximum number of gates on an input-to-output path. Boolean circuits with fan-out $1$ (i.e.\ tree-like Boolean circuits) are called {\em Boolean formulas}. (Boolean formulas are precisely the same as quantifier-free first-order formulas.)

The {\em depth-$d$ circuit/formula size} of a Boolean function $f$ is the minimum size of a depth-$d$ circuit/formula that computes $f$. {\em $\ACzero$} refers to constant-depth, $\mr{poly}(n)$-size sequences of Boolean circuits/formula on $\mr{poly}(n)$ variables. For a sequence $(f_n)$ of Boolean functions on $\mr{poly}(n)$ variables and a constant $c > 0$, we say that ``$(f_n)$ has $\ACzero$ circuit/formula size $O(n^c)$ (resp.\ $\Omega(n^c)$)'' if for some $d$ (resp.\ for all $d$), the depth-$d$ circuit/formula size of $f_n$ is $O_d(n^c)$ (resp.\ $\Omega_d(n^c)$) for all $n$.

One slightly unusual complexity measure (which arises in the descriptive complexity correspondence between $\ACzero$ and first-order logic in Section \ref{sec:prelims3}) is {\em fan-in $n$ depth}, that is, the minimum depth required to compute a Boolean function by $\ACzero$ circuits with fan-in restricted to $n$. Note that $\ACzero$ formula size lower bounds imply fan-in $n$ depth lower bounds: if $f$ has $\ACzero$ formula size $\omega(n^c)$, then its fan-in $n$ formula depth is at least $c$ (for sufficiently large $n$). (This follows from the observation that every depth-$d$ formula with fan-in $n$ is equivalent to a depth-$d$ formula of size at most $n^d$.)

\subsection{Monotone Projections}

\begin{df}[Monotone-Projection Reductions]
For Boolean functions $f : \{0,1\}^I \to \{0,1\}$ and $g : \{0,1\}^J \to \{0,1\}$, 
a {\em monotone-projection reduction} from $f$ to $g$ is a map $\rho : J \to I \cup \{0,1\}$ such that $f(x) = g(\rho^\ast(x))$ for all $x \in \{0,1\}^I$ where $\rho^\ast(x) \in \{0,1\}^J$ is defined by
\[
  (\rho^\ast(x))_j 
  = 
  \begin{cases} 
  x_i &\text{if } \rho(j) = i \in I,\\
  0 &\text{if } \rho(j) = 0,\\ 
  1 &\text{if } \rho(j) = 1.
  \end{cases}
\]
(Properly speaking, the ``reduction'' from $f$ to $g$ is the map $\rho^\ast : \{0,1\}^I \to \{0,1\}^J$ induced by $\rho$.) Notation $f \lemp g$ denotes the existence of a monotone-projection reduction from $f$ to $g$.
\end{df}

When describing monotone-projection reductions later in this paper, it will be natural to speak in terms of indexed sets of Boolean variables $\{X_i\}_{i \in I}$ and $\{Y_j\}_{j \in J}$, rather than sets $I$ and $J$ themselves. Thus, a monotone-projection reduction $\rho : J \to I \cup \{0,1\}$ associates each variable $Y_j$ with either a constant ($0$ or $1$) or some variable $X_i$.

Note that $\lemp$ is a partial order on Boolean functions. This is the simplest kind of reduction in complexity theory. It has the nice property that every standard complexity measure on Boolean functions is monotone under $\lemp$. For instance, letting $L_d(f)$ denote the depth-$d$ formula size of $f$, we have $f \lemp g \,\Longrightarrow\, L_d(f) \le L_d(g)$.

\subsection{Tree-Width and Tree-Depth} 

{\em Graphs} in this paper are finite simple graphs. (In contrast to the previous discussion of infinite structures, we assume finiteness whenever we speak of graphs.) Formally, a graph $G$ is a pair $(V(G),E(G))$ where $V(G)$ is a finite set and $E(G) \subseteq \binom{V(G)}{2}$ is a set of unordered pairs of vertices.

Four specific graphs that arise in this paper: for $k \ge 1$, let $K_k$ denote the complete graph of order $k$, let $P_k$ denote the path of order $k$, let $B_k$ denote the complete binary tree of height $k$ (where every leaf-to-root path has order $k$), and let $\mathit{Grid}_{k\times k}$ denote the $k \times k$ grid graph. (In the case $k=1$, all four of these graphs are a single vertex.)

We recall the definitions of two structural parameters, tree-width and tree-depth, which play an important role in this paper. A {\em tree decomposition} of a graph $G$ consists of a tree $T$ and a family $\mc W = \{W_t\}_{t \in V(T)}$ of sets $W_t \subseteq V(G)$ satisfying
\begin{itemize}
  \item
    $\bigcup_{t \in V(T)} W_t = V(G)$ and every edge of $G$ has both ends in some $W_t$, and
  \item
    if $t,t',t'' \in V(T)$ and $t'$ lies on the path in $T$ between $t$ and $t''$, then $W_t \cap W_{t''} \subseteq W_{t'}$.
\end{itemize}
The {\em tree-width} of $G$, denoted $\tw(G)$, is the minimum of $\max_{t \in V(T)} |W_t|-1$ over all tree decompositions $(T,\mc W)$ of $G$.

The {\em tree-depth} of $G$, denoted $\td(G)$, is the minimum height of a rooted forest $F$ such that $V(F)=V(G)$ and every edge of $G$ has both ends in some branch in $F$ (i.e.\ for every $\{v,w\} \in E(G)$, vertices $v$ and $w$ have an ancestor-descendant relationship in $F$). There is also an inductive characterization of tree-depth: if $G$ has connected components $G_1,\dots,G_t$, then
\[
  \td(G) = 
  \begin{cases}
    1 &\text{if $|V(G)| = 1$,}\\
    1+\ds\min_{v \in V(G)} \td(G - v) &\text{if }t=1\text{ and } |V(G)|>1,\\
    \ds\max_{i \in \{1,\dots,t\}} \td(G_i) &\text{if }t>1.
  \end{cases}
\]

These two structural parameters, tree-width and tree-depth, are related by inequalities:
\begin{equation}\label{eq:twtd}
  \tw(G) \le \td(G) - 1 \le \tw(G)\cdot\log|V(G)|.
\end{equation}
Tree-depth is also related the length of the longest path in $G$, denoted $\lp(G)$: 
\begin{equation}\label{eq:lptd}
  \log(\lp(G)+1) \le \td(G) \le \lp(G).
\end{equation}
(See Ch.\ 6 of \cite{nevsetril2012sparsity} for background on tree-depth and proofs of these inequalities.)\bigskip

Graph parameters $\tw(\cdot)$ and $\td(\cdot)$, as well as $\lp(\cdot)$, are easily seen to be monotone under the graph-minor relation. A reminder what this means: recall that a graph $H$ is a {\em minor} of a graph $G$, denoted $H \preceq G$, if $H$ can be obtained from $G$ by a sequence of edge contractions and vertex/edge deletions. A graph parameter $f : \{$graphs$\} \to \N$ is said to be {\em minor-monotone} if $H \preceq G \,\Longrightarrow\, f(H) \le f(G)$ for all graphs $H$ and $G$.

\section{The Colored $G$-Subgraph Isomorphism Problem}\label{sec:subgraph-isomorphism}

In this section, we introduce the colored $G$-subgraph isomorphism problem and state the known upper and lower bounds on its complexity with respect to $\ACzero$ circuits and formulas.

\begin{df}
For a graph $G$ and $n \in \N$, the {\em blow-up} $G^{{\uparrow}n}$ is the graph defined by
\begin{align*}
  V(G^{{\uparrow}n}) &= V(G) \times [n],\\
  E(G^{{\uparrow}n}) &= 
  \big\{\{(v,a),(w,b)\} : \{v,w\} \in E(G),\, a,b \in [n]\big\}.
\intertext{For $\alpha \in [n]^{V(G)}$, let $G^{(\alpha)}$ denote the subgraph of $G^{{\uparrow}n}$ defined by}
  V(G^{(\alpha)}) &= \big\{(v,\alpha_v) : v \in V(G)\big\},\\
  E(G^{(\alpha)}) &= \big\{\{(v,\alpha_v),(w,\alpha_w)\} : \{v,w\} \in E(G)\big\}.
\end{align*}
(Note that each $G^{(\alpha)}$ is an isomorphic copy of $G$.)
\end{df}

\begin{df}
For any fixed graph $G$, the {\em colored $G$-subgraph isomorphism problem} asks, given a subgraph $X \subseteq G^{{\uparrow}n}$, to determine whether or not there exists $\alpha \in [n]^{V(G)}$ such that $G^{(\alpha)} \subseteq X$. For complexity purposes, we view this problem as a Boolean function $\SUB{G,n} : \{0,1\}^{|E(G)|{\cdot}n^2} \to \{0,1\}$ with variables $\{X_e\}_{\smash{e \in E(G^{{\uparrow}n})}}$. We write $\SUB{G}$ for the sequence of Boolean functions $\{\SUB{G,n}\}_{n \in \N}$.
\end{df}

\subsection{Minor-Monotonicity}
 
The following observation appears in \cite{li2014ac0}.

\begin{prop}\label{prop:sub-sub}
If $H$ is a minor of $G$, then $\SUB{H} \lemp \SUB{G}$ (i.e.\ $\SUB{H,n} \lemp \SUB{G,n}$ for all $n \in \N$).
\end{prop}

\begin{proof}
By transitivity of $\lemp$, it suffices to consider the two cases where $H$ is obtained from $G$ via deleting or contracting a single edge $\{v,w\} \in E(G)$. In both cases, the monotone projection maps each variable $X_{\{(v',a),(w',b)\}}$ of $\SUB{G}$ with $\{v',w'\} \ne \{v,w'\}$ to the correspond variable $Y_{\{(v',a),(w',b)\}}$ of $\SUB{H}$. In the deletion case, we set the variable $X_{\{(v,a),(w,b)\}}$ to the constant $1$ for all $a,b \in [n]$. In the contraction case, we set $X_{\{(v,a),(w,b)\}}$ to $1$ if $a=b$ and to $0$ if $a \ne b$. (This ``planted perfect matching'' has the effect of gluing the $v$-fibre and the $w$-fibre for instances of $\SUB{H}$.) 
\end{proof}

Proposition \ref{prop:sub-sub} implies that the graph parameter $G \mapsto \mu(\SUB{G})$ is minor-monotone for any standard complexity measure $\mu : \{$Boolean functions$\} \to \N$ (e.g.\ depth-$d$ $\ACzero$ formula size). It also implies:

\begin{cor}\label{cor:path-sub}
For all graphs $G$, $\SUB{P_{\td(G)}} \lemp \SUB{G}$.
\end{cor}

\begin{proof}
Recall that $\td(G) \le \lp(G)$ by inequality (\ref{eq:lptd}). That is, every graph $G$ contains a path of length $\td(G)$.\footnote{This fact is straightforward to prove. Consider the case that $G$ is connected. Starting at any vertex of $G$, constructed a rooted tree $T$ by a depth-first search. Observe that for every edge $\{v,w\} \in E(G)$, it must be the case that $v$ and $w$ lie in a common branch of $T$. Therefore, the height of $T$ is an upper bound on $\td(G)$. On the other hand, note that each root-to-leaf branch of $T$ is a path in $G$. Therefore, the height of $T$ is a lower bound on $\lp(G)$.} Since subgraphs are minors, we have $P_{\td(G)} \preceq G$ and therefore $\SUB{P_{\td(G)}} \lemp \SUB{G}$ by Proposition \ref{prop:sub-sub}.
\end{proof}

\subsection{Upper Bounds}

The obvious ``brute-force'' way of solving $\SUB{G}$ has running time $O(n^{|V(G)|})$: given an input $X \subseteq G^{{\uparrow}n}$, check if $G^{(\alpha)} \subseteq X$ for each $\alpha \in [n]^{V(G)}$. A better upper bound comes from tree-width: based on an optimal tree-decomposition $(T,\mc W)$, there is a dynamic-programming algorithm with running time $n^{\tw(G)+O(1)}$ \cite{plehn1990finding}. This algorithm can be implemented by $\ACzero$ circuits of size $n^{\tw(G)+O(1)}$ and depth $O(|V(G)|)$.\footnote{It may be possible to achieve running times of $n^{\delta\cdot\tw(G)+O(1)}$ for constants $\delta < 1$ using fast matrix multiplication algorithms (cp.\ \cite{williams2014faster}). However, these algorithms appear to require logarithmic-depth circuits. For unrestricted Boolean circuits, no upper bound better than $n^{O(\tw(G))}$ is known, and in fact Marx \cite{marx2010can} has shown that the Strong Exponential Time Hypothesis rules out circuits smaller than $n^{O(\tw(G)/\log \tw(G))}$.}

Unlike circuits, formulas cannot faithfully implement dynamic-programming algorithms. The fastest known formulas for $\SUB{G}$ are tied to tree-depth: based on a minimum-height rooted forest $F$ witnessing $\td(G)$, there are $\ACzero$ formulas of size $n^{\td(G)+O(1)}$ solving $\SUB{G}$ (which come from $\ACzero$ circuits of depth $\td(G)+O(1)$ and fan-in $O(n)$). For future reference, these upper bounds are stated in the following proposition.\footnote{For the {\em uncolored} $G$-subgraph isomorphism graph, one gets essentially the same upper bounds via a reduction to $\SUB{G}$ using the ``color-coding'' technique of Alon, Yuster and Zwick \cite{AYZ95}. Amano \cite{Amano10} observed that this uncolored-to-colored reduction can be implemented by $\ACzero$ circuits.}

\begin{prop}\label{prop:upper-bounds}
For all graphs $G$, $\SUB{G}$ is solvable by $\ACzero$ circuits of size $n^{\tw(G)+O(1)}$, as well as by $\ACzero$ formulas of size $n^{\td(G) + O(1)}$.
\end{prop}

\subsection{Lower Bounds: $\ACzero$ Circuit Size}

Previous work of the author \cite{rossman2008constant} showed that the $\ACzero$ circuit size of $\SUB{K_k}$ (a.k.a.\ the (colored) $k$-CLIQUE problem) is $n^{\Omega(k)}$ for every $k \in \N$. Generalizing the technique of \cite{rossman2008constant}, Amano \cite{Amano10} gave a lower bound on the $\ACzero$ circuit size of $\SUB{G}$ for arbitrary graphs $G$. In particular, he showed that the $\ACzero$ circuit size of $\SUB{\mathit{Grid}_{k \times k}}$ is $n^{\Omega(k)}$. This result, combined with the recent Polynomial Grid-Minor Theorem\footnote{This states every graph $G$ of tree-width $k$ contains an $\Omega(k^\eps)\times\Omega(k^\eps)$ grid minor for an absolute constant $\eps > 0$.} of Chekuri and Chuznoy \cite{chekuri2014polynomial}, implies that the $\ACzero$ circuit size of $\SUB{G}$ is $n^{\Omega(\tw(G)^\eps)}$ for an absolute constant $\eps > 0$. An even stronger lower bound was subsequently proved by Li, Razborov and Rossman \cite{li2014ac0} (without appealing to the Polynomial Grid-Minor Theorem).

\begin{thm}\label{thm:LRR}
For all graphs $G$, the $\ACzero$ circuit size of $\SUB{G}$ is $n^{\Omega(\tw(G)/\log\tw(G))}$.
\end{thm}

This result is nearly tight, as it matches the upper bound of Proposition \ref{prop:upper-bounds} up to the $O(\log\tw(G))$ factor in the exponent.

\subsection{Lower Bounds: $\ACzero$ Formula Size}

For the main result of this paper (Theorem \ref{thm:improved}), we require a lower bound on the $\ACzero$ formula size of $\SUB{G}$ (or in fact on the fan-in $O(n)$ depth of $\SUB{G}$).  Since formulas are a subclass of circuits, Theorem \ref{thm:LRR} implies that the $\ACzero$ formula size of $\SUB{G}$ is at least $n^{\Omega(\tw(G)/\log\tw(G))}$.  However, this does not match the $n^{\td(G)+O(1)}$ lower bound of Proposition \ref{prop:upper-bounds}, since $\td(G)$ may be larger than $\tw(G)$ (by up to a $\log|V(G)|$ factor).  In particular, the path $P_k$ has tree-width $1$ and tree-depth $\lceil \log(k+1) \rceil$.  Although Theorem \ref{thm:LRR} gives no non-trivial lower bound on the $\ACzero$ formula size of $\SUB{P_k}$, a nearly optimal lower bound was proved in different work of the author \cite{rossman2014formulas}:

\begin{thm}\label{thm:Pk}
The $\ACzero$ formula size of $\SUB{P_k}$ is $n^{\Omega(\log k)}$. More precisely, the \mbox{depth-$d$} formula size of $\SUB{P_k,n}$ is $n^{\Omega(\log k)}$ for all $k,d,n \in \N$ with $k \le \log\log n$ and $d \le \smash{\frac{\log n}{(\log\log n)^3}}$.
\end{thm}

Via the relationship between $\ACzero$ formula size and fan-in $O(n)$ circuit depth, Theorem \ref{thm:Pk} implies:

\begin{cor}\label{cor:fanin-depth}
Circuits with fan-in $O(n)$ computing $\SUB{P_k}$ have depth $\Omega(\log k)$.
\end{cor}

\begin{rmk}
We mention a few other lower bounds related to Corollary \ref{cor:fanin-depth}.  A recent paper of Chen, Oliveira, Servedio and Tan \cite{chen2015near} gives a nearly optimal size-depth trade-off for $\ACzero$ circuits computing $\SUB{P_k}$. Namely, they prove that the depth-$d$ circuit size of $\SUB{P_k,n}$ is $n^{\Omega(d^{-1}k^{1/(d-1)})}$ for all $k \le n^{1/5}$. (This result is incomparable to Theorem \ref{thm:Pk}.) As a corollary, this shows that circuits with fan-in $O(n)$ computing $\SUB{P_k}$ have depth $\Omega(\log k/\log\log k)$ (a slightly weaker bound than Corollary \ref{cor:fanin-depth}). Previous size-depth trade-offs due to Beame, Impagliazzo and Pitassi \cite{beame1998improved} and Ajtai \cite{ajtai1989first} imply lower bounds of $\Omega(\log\log k)$ and $\Omega(\log^\ast k)$ respectively on the fan-in $O(n)$ depth of $\SUB{P_k}$.
\end{rmk}

In Section \ref{sec:first-bound}, we use Corollary \ref{cor:fanin-depth} (together with Corollary \ref{cor:path-sub}) to prove a weak version of our main result, Theorem \ref{thm:improved}, with an exponential upper bound $\beta(k) \le 2^{O(k)}$ on the quantifier-rank blow-up. We remark that the lower bound of Chen et al.\ implies a slightly weaker upper bound of $k^{O(k)}$, while the very first non-trivial lower bound of Ajtai implies a {\em non-elementary} upper bound on $\beta(k)$ (similar to the original proof of Theorem \ref{thm:finiteHPT}). For the polynomial upper bound $\beta(k) \le k^{O(1)}$, we require a stronger $n^{\Omega(\td(G)^\eps)}$ lower bound on the $\ACzero$ formula size of $\SUB{G}$ for arbitrary graphs $G$, as we explain in Section \ref{sec:extension}.

\section{Preliminaries, III}\label{sec:prelims3}

In this section, we state a few needed lemmas on the relationship between first-order logic and $\ACzero$ formula size.  As before, let $\sigma$ be a fixed finite relational signature. However, we now stipulate that {\bf\em all structures in Sections \ref{sec:prelims3} and \ref{sec:improved} are finite}. That is, we drop the adjective ``finite'' everywhere since it is assumed. Asymptotic notation in these sections ($O(\cdot)$, etc.)\ implicitly depends on $\sigma$ (although, essentially without loss of generality, it suffices to prove our results in the special case $\sigma = \{R^{(2)}\}$ of a single binary relation).

\subsection{Descriptive Complexity: $\FO = \ACzero$}

\begin{df}[Gaifman Graphs, Encodings, $\MODEL_\Phi$]\
\begin{itemize}
\item
For a structure $\mf A$, we denote by $\Gaif(\mf A)$ the {\em Gaifman graph} of $\mf A$. This is the graph whose vertex set is the universe of $\mf A$ and whose edges are pairs $\{v,w\}$ such that $v \ne w$ and $v,w$ appear together in a tuple of any relation of $\mf A$.
\item
If $\mf A$ has universe $[n]$, then we denote by $\mr{Enc}(\mf A) \in \{0,1\}^{\wh n}$ the standard bit-string encoding of $\mf A$ where $\wh n = \sum_{R^{(t)} \in \sigma} n^t$ ($=n^{O_\sigma(1)}$). That is, each bit of $\mr{Enc}(\mf A)$ is the indicator for a tuple of some relation of $\mf A$. (Note that $\mr{Enc}(\cdot)$ is a bijection between structures with universe $[n]$ and strings in $\{0,1\}^{\wh n}$.)  
\item
For a first-order sentence $\Phi$ and $n \in \N$, let $\MODEL_{\Phi,n} : \{0,1\}^{\wh n} \to \{0,1\}$ be the Boolean function defined, for structures $\mf A$ with universe $[n]$, by 
\[
  \MODEL_{\Phi,n}(\mr{Enc}(\mf A)) = 1 \defiff \mf A \models \Phi.\qquad
\]
We write $\MODEL_\Phi$ for the sequence of Boolean functions $\{\MODEL_{\Phi,n}\}_{n \in \N}$.
\end{itemize}
\end{df}

The next lemma gives one-half of the descriptive complexity correspondence between first-order logic and $\ACzero$:

\begin{la}[``$\FO \subseteq \ACzero$'']\label{la:FO-AC0}
For all $1 \le w \le k$, if $\Phi$ is a first-order sentence of quantifier-rank $k$ and variable-width $w$, then $\MODEL_\Phi$ is computable by $\ACzero$ circuits of depth $k$ and fan-in $O(n)$ and size $O(n^w)$. These circuits are equivalent with $\ACzero$ formulas of depth $k$ and size $O(n^k)$.
\end{la}

(To be completely precise, each of these $O(\cdot)$ terms is really $O_{\sigma,k}(\cdot)$, that is, with constants that depend on $k$ as well as the signature $\sigma$.) We remark that Lemma \ref{la:FO-AC0} has a converse (``$\ACzero \subseteq \FO$'') with respect to both the uniform and non-uniform versions of $\ACzero$. We omit the statement of these results, since the description of $\ACzero$ circuits via first-order sentences is not needed in this paper (see \cite{Imm} for details).

\subsection{Retracts, Cores, Hom-Preserved Classes}

The last bit of required background concerns homomorphism-preserved classes of structures. We begin by defining the key notions of homomorphic equivalence and cores.

\begin{df}[Homomorphic Equivalence, (Co-)Retracts, Cores]\
\begin{itemize}
\item
Recall notation $\mf A \to \mf B$ denoting the existence of a homomorphism from $\mf A$ to $\mf B$.
\item
Structures $\mf A$ and $\mf B$ are {\em homomorphically equivalent}, denoted $\mf A \bihom \mf B$, if $\mf A \to \mf B$ and $\mf B \to \mf A$.
\item
We write $\mf A \retr \mf B$ and say that $\mf B$ is a {\em retract} of $\mf A$ and $\mf A$ is a {\em co-retract} of $\mf B$ if: (1) $\mf B$ is a substructure of $\mf A$ and (2) there exists a homomorphism $\mf A \to \mf B$ that fixes $\mf B$ pointwise (a.k.a.\ a retraction). (Note that $\mf A \retr \mf B$ implies $\mf A \bihom \mf B$.)
\item
A structure $\mf A$ is a {\em core} if it has no proper retract (that is, $\mf A \retr \mf B \,\Longrightarrow\, \mf A = \mf B$).
\end{itemize}
\end{df}

The next lemma states a few basic properties of cores (see \cite{HellNesCore,HellNes}).

\begin{la}\label{la:cores}\
\begin{enumerate}[\quad\normalfont(a)]
  \item
    Every $\bihom$-equivalence class contains a unique core up to isomorphism. (That is, every structure $\mf A$ is homomorphically equivalent to a unique core.)
  \item
    \label{la:finitely-many}
    For every $k$, there are only finitely many non-isomorphic cores of tree-depth $k$. (This number depends on the signature $\sigma$.)
  \item
    (As an aside:) If a graph $G$ is a core, then the colored and uncolored $G$-subgraph isomorphism problems are equivalent under linear-size monotone-projection reductions (see \cite{Grohe07,li2014ac0}).
\end{enumerate}
\end{la}

\begin{df}[Hom-Preserved Classes, Minimal Cores]\
\begin{itemize}
\item
We say that a class of structures $\scr C$ (i.e.\ a class of finite structures)
is {\em hom-preserved} if, whenever $\mf A \in \scr C$ and $\mf A \to \mf B$, we have $\mc B \in \scr C$.
\item
For a hom-preserved class $\scr C$, let $\MinStruct(\scr C)$ be the set of $\mf M \in \scr C$ with the property that for all structures $\mf A$, if $\mf A \in \scr C$ and $\mf A \to \mf M$, then $\mf M$ is isomorphic to a retract of $\mf A$.
\end{itemize}
\end{df}

The next lemma states the essential properties of $\MinStruct(\scr C)$ (see \cite{rossman2008homomorphism}).

\begin{la}\label{la:MinStruct}
The following hold for any hom-preserved class $\scr C$:
\begin{enumerate}[\quad\normalfont(a)]
\item
\label{la:easy}
$\mf A \in \scr C$ if, and only if, there exists $\mf M \in \MinStruct(\scr C)$ such that $\mf M \to \mf A$.
\item
Every structure in $\MinStruct(\scr C)$ is, indeed, a core. 
\item\label{la:iso}
Every homomorphism between structures in $\MinStruct(\scr C)$ is an isomorphism. 
\item
\label{la:qr-td}
$\scr C$ is definable (i.e.\ within the class of all finite structures) by an existential-positive sentence of quantifier-rank $k$ if, and only if, $\td(\Gaif(\mf M)) \le k$ for all $\mf M \in \MinStruct(\scr C)$.
\item
\label{la:vw-tw}
$\scr C$ is definable by an existential-positive sentence of variable-width $w$ if, and only if, $\MinStruct(\scr C)$ contains finitely many non-isomorphic structures and $\tw(\Gaif(\mf M)) \le w$ for every $\mf M \in \MinStruct(\scr C)$.
\end{enumerate}
\end{la}

Since Lemma \ref{la:MinStruct}(\ref{la:qr-td}) in particular plays a key role in the next section, we briefly sketch the proof. In one direction: Suppose $\scr C$ is defined by an existential-positive sentence $\Phi$ of quantifier-rank $k$. It is easy to show (by a syntactic argument) that $\Phi$ is equivalent to a disjunction $\Psi_1 \vee \dots \vee \Psi_t$ of {\em primitive-positive sentences} $\Psi_i$ (i.e.\ existential-positive sentences that involve conjunctions $\wedge$ but no disjunctions $\vee$), each with quantifier-rank at most $k$. For each $\Psi_i$, there is a corresponding structure $\mc A_i$ with the property that $\mc B \models \Psi_i \,\Leftrightarrow\, \mc A_i \to \mc B$ and moreover the tree-depth of $\mc A_i$ is at most the quantifier-rank of $\Psi_i$ (and hence at most $k$). Thus, $\scr C$ is generated by $\mc A_1,\dots,\mc A_t$ and hence $\MinStruct(\scr C)$ consists of finitely many cores, each of tree-depth at most $k$ (coming from the minimal elements among $\mc A_1,\dots,\mc A_t$ in the homomorphism order). 

For the reverse direction: Start with the assumption that all structures in $\MinStruct(\scr C)$ have tree-depth at most $k$. By Lemma \ref{la:cores}(\ref{la:finitely-many}), $\MinStruct(\scr C)$ contains finitely many non-isomorphic structures $\mc M_1,\dots,\mc M_t$. For each $\mc M_i$, let $\Psi_i$ be the corresponding primitive-positive sentence of quantifier-rank at most $k$. Then $\scr C$ is defined by the existential-positive sentence $\Psi_1 \vee \dots \vee \Psi_t$.

\section{Proof of Theorem \ref{thm:improved}}\label{sec:improved}

In this section, we finally prove our main result, the ``Poly-rank'' Homomorphism Preservation Theorem on Finite Structures (Theorem \ref{thm:improved}, stated in Section \ref{sec:preservation-theorems}). We begin in Section \ref{sec:first-bound} by proving a weaker version of the result with an exponential upper bound $\beta(k) \le 2^{O(k)}$. In Section \ref{sec:extension}, we describe the improvement to $\beta(k) \le k^{O(1)}$, which involves new results from circuit complexity and graph minor theory.

\subsection{Preliminary Bound: $\beta(k) \le 2^{O(k)}$}\label{sec:first-bound}

For simplicity sake, we will assume that $\sigma$ consists of binary relations only. At the end of this subsection, we explain how to extend the argument to arbitrary $\sigma$.

Let $\Phi$ be a first-order sentence of quantifier-rank $k$, let $\scr C$ be the set of finite models of $\Phi$, and assume that $\scr C$ is hom-preserved (that is, $\Phi$ is preserved under homomorphisms on finite structures). Our goal is to show that $\Phi$ is equivalent to an existential-positive sentence of quantifier-rank $2^{O(k)}$. By Lemma \ref{la:MinStruct}(\ref{la:qr-td}), it suffices to show that $\td(\Gaif(\mf M)) \le 2^{O(k)}$ for all $\mf M \in \MinStruct(\scr C)$. 

Consider any $\mf M \in \MinStruct(\scr C)$. Let $G$ be the Gaifman graph of $\mf M$, and let $m$ be the size of the universe of $\mf M$. (Note that $m = |V(G)|$.) The following claim is key to showing $\td(G) \le 2^{O(k)}$.

\begin{claim}\label{claim:proj}
For all $n \in \N$, there exists a monotone-projection reduction $\SUB{G,n} \lemp \MODEL_{\Phi,mn}$.
\end{claim}

In order to define this monotone-projection reduction, let us identify $[mn]$ with the set $V(G^{{\uparrow}n})$ ($=V(G) \times [n]$). Variables $X_e$ of $\SUB{G,n}$ are indexed by potential edges $e \in E(G^{{\uparrow}n})$ in a subgraph $X \subseteq G^{{\uparrow}n}$. Variables $Y_i$ of $\MODEL_{\Phi,mn}$ are indexed by the set
\[
  I \defeq \big\{(R,(v,a),(w,b)) : R^{(2)} \in \sigma,\, (v,a),(w,b) \in V(G^{{\uparrow}n})\big\}.
\]
(That is, $I$ is the set of potential $2$-tuples of relations of structures with universe $V(G^{{\uparrow}n})$.) Define the monotone projection $\rho : \{Y_i\}_{i \in I} \to \{X_e\}_{\smash{e \in E(G^{{\uparrow}n})}} \cup \{0,1\}$ by
\[
  \rho : Y_{(R,(v,a),(w,b))} \mapsto
  \begin{cases}
    X_{\{(v,a),(w,b)\}} &\text{if } (v,w) \in R^{\mf M} \text{ and } v \ne w,\\
    1 &\text{if } (v,w) \in R^{\mf M} \text{ and } v = w,\\
    0 &\text{otherwise}.
  \end{cases}
\]
We must show that the corresponding map 
\[
  \rho^\ast : \{\text{subgraphs of }G^{{\uparrow}n}\} \to \{\text{structures with universe }V(G^{{\uparrow}n})\}
\]
is in fact a reduction from $\SUB{G,n}$ to $\MODEL_{\Phi,mn}$. That is, we must show that for any $X \subseteq G^{{\uparrow}n}$,
\begin{equation}\label{eq:to-show}
  \SUB{G,n}(X) = 1 \,\Longleftrightarrow\, \MODEL_{\Phi,mn}(\rho^\ast(X)) = 1.
\end{equation}

For the $\Longrightarrow$ direction of (\ref{eq:to-show}): Assume $\SUB{G,n}(X) = 1$. Then $G^{(\alpha)} \subseteq X$ for some $\alpha \in [n]^{V(G)}$. The definition of $\rho$ ensures that the map $v \mapsto (v,\alpha_v)$ is a homomorphism from $\mf M$ to the structure $\rho^\ast(X)$.\footnote{To see why, suppose we have $(v,w) \in R^{\mf M}$ for some $R^{(2)} \in \sigma$. We must show that $((v,\alpha_v),(w,\alpha_w)) \in R^{\rho^\ast(X)}$. First, consider the case that $v \ne w$. The assumption $G^{(\alpha)} \subseteq X$ implies that $\{(v,\alpha_v),(w,\alpha_w)\} \in E(X)$. Since $\rho$ maps the variable $Y_{(R,(v,\alpha_v),(w,\alpha_w))}$ to the variable $X_{\{(v,\alpha_v),(w,\alpha_w)\}}$ (which has value $1$ for $X$), it follows that $((v,\alpha_v),(w,\alpha_w)) \in R^{\rho^\ast(X)}$. Finally, consider the case that $v = w$. In this case, $\rho$ maps the variable $Y_{(R,(v,\alpha_v),(w,\alpha_w))}$ to the constant $1$. So again we have $((v,\alpha_v),(w,\alpha_w)) \in R^{\rho^\ast(X)}$.} Since $\scr C$ is hom-preserved, it follows that $\rho^\ast(X) \in \scr C$ and therefore $\MODEL_{\Phi,mn}(\rho^\ast(X)) = 1$.

For the $\Longleftarrow$ direction of (\ref{eq:to-show}): Assume $\MODEL_{\Phi,mn}(\rho^\ast(X)) = 1$, that is, $\rho^\ast(X) \in \scr C$. By Lemma \ref{la:MinStruct}(\ref{la:easy}) there exist $\mf N \in \MinStruct(\scr C)$ and a homomorphism $\gamma : \mf N \to \rho^\ast(X)$. The definition of $\rho$ ensures that the map $\pi : (v,i) \mapsto v$ is a homomorphism from $\rho^\ast(X)$ to $\mf M$.\footnote{In fact, this holds for every $X \subseteq G^{{\uparrow}n}$ independent of the assumption that $\MODEL_{\Phi,mn}(\rho^\ast(X)) = 1$. This follows from the observation that $\pi$ is (in particular) a homomorphism from $\rho^\ast(G^{{\uparrow}n})$ to $\mf M$. To see why, consider any $((v,a),(w,b)) \in R^{\rho^\ast(G^{{\uparrow}n})}$ (corresponding to $Y_{(R,(v,a),(w,b))} = 1$). It must be the case that $(v,w) \in R^{\mf M}$, since the contrary assumption $(v,w) \notin R^{\mf M}$ would mean that $\rho$ maps the variable $Y_{(R,(v,a),(w,b))}$ to $0$.} The composition $\pi \circ \gamma$ is a homomorphism from $\mf N$ to $\mf M$. By Lemma \ref{la:MinStruct}(\ref{la:iso}), it is an isomorphism. Therefore, without loss of generality, we may assume that $\mf M = \mf N$ and $\pi \circ \gamma$ is the identity map on the universe $V(G)$ of $\mf M$. This means that $\pi(v) \in \{(v,a) : a \in [n]\}$ for all $v \in V(G)$. We may now define $\alpha \in [n]^{V(G)}$ as the unique element such that $\gamma : v \mapsto (v,\alpha_v)$ for all $v \in V(G)$. From the definition of $\rho$ and the fact that $G = \Gaif(\mf M)$, we infer that $G^{(\alpha)} \subseteq X$.\footnote{Consider an edge $\{(v,\alpha_v),(w,\alpha_v)\} \in E(G^{(\alpha)})$. By definition of $G^{(\alpha)}$, we have $\{v,w\} \in E(G)$. Since $G = \Gaif(\mf M)$, there exists a relation $R^{(2)} \in \sigma$ such that $(v,w) \in R^{\mf M}$ or $(w,v) \in R^{\mf M}$. Without loss of generality, assume $(v,w) \in R^{\mf M}$. Since $\gamma : \mf M \to \rho^\ast(X)$ is a homomorphism, we have $(\gamma(v),\gamma(w)) = ((v,\alpha_v),(w,\alpha_w)) \in R^{\rho^\ast(X)}$. Since $(v,w) \in R^{\mf M}$ and $v \ne w$, the monotone projection $\rho$ maps $Y_{(R,(v,\alpha_v),(w,\alpha_w))}$ to $X_{\{(v,\alpha_v),(w,\alpha_w)\}}$. It follows that $\{(v,\alpha_v),(w,\alpha_w)\} \in E(X)$. Therefore, $G^{(\alpha)}$ is a subgraph of $X$.} We conclude that $\SUB{G,n}(X) = 1$, finishing the proof of Claim \ref{claim:proj}.
\bigskip

We proceed to show that $\td(G) \le 2^{O(k)}$.  By Corollary \ref{cor:path-sub}, we have $\SUB{P_{\td(G)},n} \lemp \SUB{G,n}$. By Claim \ref{claim:proj} and transitivity of $\lemp$, it follows that $\SUB{P_{\td(G)},n} \lemp \MODEL_{\Phi,kn}$. Therefore, $\mu(\SUB{P_{\td(G)},n}) \le \mu(\MODEL_{\Phi,kn})$ for every standard complexity measure $\mu : \{$Boolean functions$\} \to \N$ (in particular, depth-$k$ formula size). By Lemma \ref{la:FO-AC0} (the simulation of first-order logic by $\ACzero$), there exist depth-$k$ formulas of size $O((mn)^k)$ which compute $\MODEL_{\Phi,mn}$. Therefore, there exist depth-$k$ formulas of size $O((mn)^k)$ which compute $\SUB{P_{\td(G)},n}$. On the other hand, by Theorem~\ref{thm:Pk}, the depth-$k$ formula size of $\SUB{P_{\td(G)},n}$ is $n^{\Omega(\log\td(G))}$ for all sufficiently large $n$ such that $k < \log\log n$. Therefore, we have $n^{\Omega(\log\td(G))} \le O((mn)^k)$ for all sufficiently large $n$. Since $m$ ($= |V(G)|$) is constant, it follows that $k \ge \Omega(\log\td(G))$, that is, $\td(G) \le 2^{O(k)}$. This completes the proof that $\beta(k) \le 2^{O(k)}$ for binary signatures $\sigma$.

\begin{rmk}
In this argument, as an alternative to {\em depth-$k$ formula size}, we may instead consider {\em fan-in $O(n)$ depth} (i.e.\ {\em fan-in $cn$ depth} for a sufficiently large constant $c$) and appeal to Corollary \ref{cor:fanin-depth} instead of Theorem \ref{thm:Pk}.
\end{rmk}

Finally, we explain how to adapt the above argument when $\sigma$ is an arbitrary finite relational signature. Here the variables of $\MODEL_{\Phi,kn}$ are indexed by the set
\[
  \big\{(R,(v_1,a_1),\dots,(v_t,a_t)) : R^{(t)} \in \sigma,\, (v_1,a_1),\dots,(v_t,a_t) \in V(G) \times [n]\big\}
\]
and the reduction $\{\text{subgraphs of }G^{{\uparrow}n}\} \to \{\text{structures with universe }V(G^{{\uparrow}n})\}$ is defined by
\[
  Y_{(R,(v_1,a_1),\dots,(v_t,a_t))} \mapsto
  \begin{cases}
    \bigwedge_{1 \le i < j \le t \,:\, v_i \ne v_j} 
    X_{\{(v_i,a_i),(v_j,a_j)\}}
    &\text{if } (v_1,\dots,v_t) \in R^{\mf M},\\
    0 &\text{otherwise}.
  \end{cases}
\]
(By convention, $\bigwedge_{i \in \emptyset} X_i = 1$.) Note that this reduction is not a monotone projection, as we are mapping each $Y$-variable to a conjunction of $X$-variables. This reduction is, however, computed by a single layer of constant fan-in AND gates. Therefore, under this reduction, any Boolean formula computing $\MODEL_{\Phi,kn}$ is converted to a Boolean formula computing $\SUB{G,n}$ with an increase of $1$ in depth and a constant factor increase in size. Other than this change, the rest of the argument is identical to the case of binary signatures.

\subsection{Improvement to $\beta(k) \le k^{O(1)}$}\label{sec:extension}

The upper bound $\beta(k) \le 2^{O(k)}$ in the previous section relies on the {\em exponential approximation} of tree-depth in terms of the longest path, that is, $\log(\lp(G)+1) \le \td(G) \le \lp(G)$ (inequality (\ref{eq:lptd})). To achieve a polynomial upper bound on $\beta(k)$, we require a {\em polynomial approximation} of tree-depth in terms of a few manageable classes ``excluded minors''. This realization led to a conjecture of the author, which was soon proved in joint work with Ken-ichi Kawarabayashi \cite{RossmanTREEDEPTH}.

\begin{thm}\label{thm:treedepth}
Every graph $G$ of tree-depth $k$ satisfies one (or more) of the following conditions for $\ell = \wt\Omega(k^{1/5})$:
\begin{enumerate}[\ \ \normalfont(i)]
  \item
    $\tw(G) \ge \ell$,
  \item
    $G$ contains a path of length $2^\ell$, or
  \item
    $G$ contains a $B_\ell$-minor.
\end{enumerate}
\end{thm}

This result is analogous to the Polynomial Grid-Minor Theorem \cite{chekuri2014polynomial}, which can be used to replace condition (i) with the condition that $G$ contains an $\Omega(k^\eps) \times \Omega(k^\eps)$ grid minor for an absolute constant $\eps > 0$. In cases (i) and (ii), Theorems \ref{thm:LRR} and \ref{thm:Pk} respectively imply that $\SUB{G}$ has $\ACzero$ formula size $n^{\wt\Omega(\td(G)^{1/5})}$. This leaves only case (iii), where forthcoming work of the author \cite{RossmanSUBGRAPH} shows the following (via a generalization of the ``pathset complexity'' framework of \cite{rossman2014formulas}).

\begin{thm}\label{thm:Bk}
The $\ACzero$ formula size of $\SUB{B_k}$ is $n^{\Omega(k^\eps)}$ for an absolute constant $\eps > 0$.
\end{thm}

Together Theorems \ref{thm:treedepth} and \ref{thm:Bk} imply:

\begin{thm}\label{thm:td-bound}
For all graphs $G$, the $\ACzero$ formula complexity of $\SUB{G}$ is $n^{\Omega(\td(G)^\eps)}$ for an absolute constant $\eps > 0$.
\end{thm}

Plugging Theorem \ref{thm:td-bound} into the argument in the previous subsection directly yields the polynomial upper bound $\beta(k) \le k^{O(1)}$ of Theorem \ref{thm:finiteHPT}. (In fact, we get $\beta(k) \le k^{1/\eps}$ for the constant $\eps > 0$ of Theorem \ref{thm:td-bound}.)

\section{Comparison with the Method in (R.\ 2008)}\label{sec:comparison}
 
In this section, for the sake of comparison, we summarize the model-theoretic approach of the original proof of Theorem \ref{thm:finiteHPT} in \cite{rossman2008homomorphism}.  The starting point in \cite{rossman2008homomorphism} is a new compactness-free proof of the classical Homomorphism Preservation Theorem, which moreover yields the stronger ``equi-rank'' version (Theorem \ref{thm:equirank}).  The proof is based on an operation mapping each structure $\mc A$ to an infinite co-retract $\Gamma(\mc A)$. (We drop the assumption of the last two sections that structures are finite by default.) In order to state the key property of this operation, we introduce notation $\mf A \equiv_{\mr{FO}(k)} \mf B$ (resp.\ $\mf A \equiv_{\exists^+\mr{FO}(k)} \mf B$) denoting the statement that $\mf A$ and $\mf B$ satisfy the same first-order sentences (resp.\ existential-positive sentences) of quantifier-rank $k$.

\begin{thm}\label{thm:Gamma}
There is an operation $\Gamma : \{$structures$\} \to \{$structures$\}$ associating every structure $\mf A$ with a co-retract $\Gamma(\mf A) \retr \mf A$ such that, for all structures $\mf A$ and $\mf B$ and $k \in \N$,
\[
  \mf A \equiv_{\exists^+\FO(k)} \mf B \ \Longrightarrow\ \Gamma(\mf A) \equiv_{\FO(k)} \Gamma(\mf B).
\] 
\end{thm}

There is a straightforward proof that Theorem \ref{thm:Gamma} implies Theorem \ref{thm:equirank} (see \cite{rossman2008homomorphism}).  The structure $\Gamma(\mf A)$ is the Fra\"{\i}sse limit of the class of co-finite co-retracts of $\mf A$ (that is, structures $\mf A'$ such that $\mf A' \retr \mf A$ and $A' \setminus A$ is finite).  We remark that $\Gamma(\mf A)$ is infinite, even when $\mf A$ is finite.  For this reason, Theorem \ref{thm:Gamma} says nothing in the setting of finite structures.

The Homomorphism Preservation Theorem on Finite Structures (Theorem \ref{thm:finiteHPT}) is proved in \cite{rossman2008homomorphism} by considering a sequence of finitary ``approximations'' of $\Gamma(\mf A)$. (This is somewhat analogous to sense in which large random graph $G(n,1/2)$ ``approximate'' the infinite Rado graph.)

\begin{thm}\label{thm:Gamma-r}
There is a computable function $\beta : \N \to \N$ and a sequence $\{\Gamma_k\}_{k \in \N}$ of operations $\Gamma_k : \{$finite structures$\} \to \{$finite structures$\}$ associating every finite structure $\mf A$ with a sequence $\{\Gamma_k(\mf A)\}_{k \in \N}$ of finite co-retracts $\Gamma_k(\mf A) \stackrel{\scriptscriptstyle\supseteq}{\to} \mf A$ such that, for all finite structures $\mf A$ and $\mf B$ and $k \in \N$,
\[
  \mf A \equiv_{\exists^+\FO(\beta(k))} \mf B \ \Longrightarrow\ \Gamma_k(\mf A) \equiv_{\FO(k)} \Gamma_k(\mf B).
\]
\end{thm}

Theorem \ref{thm:finiteHPT} follows directly from Theorem \ref{thm:Gamma-r}, inheriting the same quantifier-rank blow-up $\beta(k)$. The proof of Theorem \ref{thm:Gamma-r} in \cite{rossman2008homomorphism} implies a non-elementary upper bound on $\beta(k)$. While the present paper improves the upper bound $\beta(k) \le k^{O(1)}$ in Theorem \ref{thm:finiteHPT}, we remark that this it does not imply any improvement to $\beta(k)$ in Theorem \ref{thm:Gamma-r}.

\section{Syntax vs. Semantics in Circuit Complexity}\label{sec:conclusion}

We conclude this paper by stating some consequences of our results in circuit complexity.  Let $\mr{HomPreserved}$ denote the class of all homomorphism-preserved graph properties (for example, $\{G : \mr{girth}(G) \le 20$ or $\mr{clique\text{-}number}(G) \ge 10\}$).  This is a semantic class, akin to the class $\mr{Monotone}$ of all monotone languages.  The new proof in this paper of the Homomorphism Preservation Theorem on Finite Structures using $\ACzero$ lower bounds is easily to imply the following ``Homomorphism Preservation Theorem for (non-uniform) $\ACzero$'':
\[
  \mr{HomPreserved} \cap \ACzero 
  = \exists^+\mr{FO}\quad
  (\subseteq \{\text{poly-size monotone DNFs}\}).
\]
In other words, every homomorphism-preserved graph property in $\ACzero$ is definable (among finite graphs) by an existential-positive first-order sentence and, therefore, also by a polynomial-size monotone DNF (moreover, with constant bottom fan-in).  As a consequence, for every integer $d \ge 2$, we get a collapse of the $\ACzero$ depth hierarchy with respect to homomorphism-preserved properties:
\[
  \mr{HomPreserved} \cap \ACzero[\tu{depth }d] = \mr{HomPreserved} \cap \ACzero[\tu{depth }d+1].
\]
In contrast, it is known that $\ACzero[\tu{depth }d] \ne \ACzero[\tu{depth }d+1]$ by the Depth Hierarchy Theorem \cite{Hastad86}.

These results have an opposite nature to the ``syntactic monotonicity $\ne$ semantic monotonicity'' counterexamples of Ajtai and Gurevich \cite{AG87} and Razborov \cite{razborov1985lower} (as well as Tardos \cite{tardos1988gap}), which respectively show that
\[
  \mr{Monotone} \cap \mr{AC^0} \ne \mr{Monotone\text{-}AC^0}
  \quad\text{ and }\quad
  \mr{Monotone} \cap \mr{P} \ne \mr{Monotone\text{-}P}.
\]
In light of the results of this paper, I feel that questions of syntax vs.\ semantics in circuit complexity are worth re-examining. For instance, so far as I know, there is no known separation between the {\em uniform average-case} monotone vs.\ non-monotone complexity of any monotone function in any well-studied class of Boolean circuits ($\ACzero$, $\NCone$, etc.) It is plausible that syntactic monotonicity $=$ semantic monotonicity in the average-case. Evidence for this viewpoint comes from the considering the {\em slice distribution} (that is, the uniform distribution on inputs of Hamming weight exactly $\lfloor n/2 \rfloor$). With respect to the slice distribution, it is known that monotone and non-monotone complexity are equivalent within a $\mr{poly}(n)$ factor by a classic result of Berkowitz \cite{berkowitz1982some}.

As for an even stronger ``Homomorphism Preservation Theorem'' in circuit complexity,  we can state the following: if for every $k$, $\SUB{P_k}$ requires unbounded-depth formula size $n^{\Omega(\log k)}$ (which is widely conjectured to be true) or even $n^{\omega_{k\to\infty}(1)}$, then $\mr{HomPreserved} \cap \NCone = \exists^+\mr{FO}$. Therefore, I strongly believe in a ``Homomorphism Preservation Theorem for $\NCone$''. On the other hand, the homomorphism-preserved property of being 2-colorable a.k.a.\ non-bipartite ($=\{G : C_k \to G \text{ for any odd }k\}$) is in $\mr{Logspace}$ (this follows from Reingold's theorem \cite{reingold2008undirected}), yet it is not $\exists^+\mr{FO}$-definable. Therefore, we may assert that $\mr{HomPreserved} \cap \mr{Logspace} \ne \exists^+\mr{FO}$.

\end{document}